\documentclass[12pt]{amsart}
\usepackage{amsmath,amssymb,amsfonts,amsthm,amsopn}
\usepackage{graphicx,tikz}
\usepackage[all]{xy}
\usepackage{amsmath}
\usepackage{multirow, longtable, makecell, caption, array}
\usepackage{amssymb}
\usepackage{mathtools,enumitem}
\usepackage[utf8]{inputenc}

\newtheorem{theorem}{Theorem}[section]
\newtheorem{lemma}[theorem]{Lemma}
\newtheorem{corollary}[theorem]{Corollary}
\newtheorem{proposition}[theorem]{Proposition}

\newtheorem{remark}[theorem]{Remark}

\theoremstyle{definition}

\usepackage{cellspace} %
\def\Spdr{\mathrm{Sp}(2d,\R)}

\def\la{\lambda}

\def\cF{\mathcal{F}}
\def\cS{\mathcal{S}}

\def\cM{\mathcal{M}}

\def\cU{\mathcal{U}}
\def\cA{\mathcal{A}}
\def\cL{\mathcal{L}}

\def\cA{\mathcal{A}}
\def\bR{\mathbb{R}}

\def\bC{\mathbb{C}}

\def\rd{\bR^d}

\def\rdd{{\bR^{2d}}}

\def\R{\right)}

\def\<{\left<}
\def\>{\right>}

\def\mv1{M_v^1}

\def\R{\mathbb{R}}

\def\sch{\mathcal{S}}

\def\Fur{\mathcal{F}}

\def\Sn2{S_{2}(L^{2}(\Ren))}
\def\S1{S_{1}(L^{2}(\Ren))}
\def\sig00{\sigma_{0,0}}

\def\la{\langle}
\def\ra{\rangle}

\def\Spdr{\mathrm{Sp}(d,\bR)}
\def\Mpdr{\mathrm{Mp}(d,\bR)}

\begin{document}
	
\title[On the pointwise convergence in the Feynman-Trotter formula] {On the pointwise convergence of the integral kernels in the Feynman-Trotter formula}
\author{Fabio Nicola and S. Ivan Trapasso}
\address{Dipartimento di Scienze Matematiche,
	Politecnico di Torino, corso Duca degli Abruzzi 24, 10129 Torino,
	Italy}
\address{Dipartimento di Scienze Matematiche,
	Politecnico di Torino, corso Duca degli Abruzzi 24, 10129 Torino,
	Italy}
\email{fabio.nicola@polito.it}
\email{salvatore.trapasso@polito.it}
%\date{}

\begin{abstract}We study path integrals in the Trotter-type form for the  Schr\"odinger equation, where the Hamiltonian is the Weyl quantization of a real-valued quadratic form perturbed by
 a potential $V$ in a class encompassing that - considered by Albeverio and It\^o in celebrated papers - of Fourier transforms of complex measures. Essentially, $V$ is bounded and has the regularity of a function whose Fourier transform is in $L^1$. Whereas the strong convergence in $L^2$ in the Trotter formula, as well as several related issues at the operator norm level are well understood, the original Feynman's idea concerned the subtler and widely open problem of the pointwise convergence of the corresponding probability amplitudes, that are the integral kernels of the approximation operators. We prove that, for the above class of potentials, such a convergence at the level of the integral kernels in fact occurs, uniformly on compact subsets and for every fixed time, except for certain exceptional time values for which the kernels are in general just distributions. 
 Actually, theorems are stated for potentials in several function spaces arising in Harmonic Analysis, with corresponding convergence results. Proofs rely on Banach algebras techniques for pseudo-differential operators acting on such function spaces. 
\end{abstract}
\maketitle

\section{Introduction and main results}

The path integral formulation of Quantum Mechanics is by far one of
the major achievements in modern theoretical physics. The first intuition
on the issue is attributed to Dirac: in his celebrated 1933 paper \cite{dirac 33} he provided several clues indicating that the Lagrangian formulation of classical mechanics should have a quantum counterpart. While it is debatable whether the entire story was already known to him at the time of writing, his program has been finalised by Feynman \cite{feyn1 48}, who explicitly recognized Dirac's remarks as the main source of inspiration for his landmark contribution to a new formulation of non-relativistic quantum mechanics beyond the Schr\"odinger and Heisenberg pictures. 

\subsection{The sequential approach to path integrals}
We could argue that the path integral formulation comes from a profound understanding of the double-slit experiment - in fact, this is precisely
the way Feynman introduces the problem in the book \cite{feyn2 hibbs}.
While this is an intriguing perspective, we briefly outline this approach
from a different starting point. Recall that the state of a non-relativistic
particle in the Euclidean space $\mathbb{R}^{d}$ at time $t$ is
represented by the wave function $\psi\left(t,x\right)$, $\left(t,x\right)\in\mathbb{R}\times\mathbb{R}^{d}$, 
such that $\psi\left(t,\cdot\right)\in L^{2}(\mathbb{R}^{d})$. The dynamics under the real-valued potential
$V$ is regulated by the Cauchy problem for the Schr\"odinger equation\footnote{We set $m=1$ for the mass of the particle and $\hbar=1$
	for the Planck constant.}:
\begin{equation}
\begin{cases}
i\partial_{t}\psi=(H_{0}+V(x))\psi\\
\psi(0,x)=\varphi(x),
\end{cases}\label{cauchy schr}
\end{equation}
where $H_{0}=-\triangle/2$ is the free Hamiltonian of non-relativistic
quantum mechanics. Provided that suitable conditions on the potential
are satisfied, the dynamics in \eqref{cauchy schr} can be equivalently
recast by means of the unitary \emph{propagator} $U\left(t\right)=e^{-itH}$,
$t\in\mathbb{R}$: 
\[
\psi\left(t,x\right)=U\left(t\right)\varphi(x).
\]
At least on a formal level, one can thus represent $U\left(t\right)$
as an integral operator:
\[
\psi\left(t,x\right)=\int_{\mathbb{R}^{d}}u_t(x,x_0)\varphi(x_0)dx_0,
\]
where the kernel $u_t(x,x_0)$ intuitively yields the transition
amplitude from the position $x_0$ at time $0$ to the position $x$ at
time $t$. The path integral formulation exactly concerns the determination
of this kernel: according to Feynman's prescription, one should take
into account the many possible \emph{interfering alternative paths
}from $x_0$ to $x$ that the particle could follow. Each path would contribute to the total probability amplitude with
a phase factor proportional to the \emph{action functional} corresponding
to the path:
\[
S\left[\gamma\right]=\int_{s}^{t}L\left(\gamma\left(\tau\right),\dot{\gamma}\left(\tau\right)\right)d\tau,
\]
where $L$ is the Lagrangian of the corresponding classical system.
In short, a merely formal representation of the kernel is
\begin{equation}
u_t(x,x_0)=\int e^{iS\left[\gamma\right]}\mathcal{D}\gamma,\label{prop path int}
\end{equation}
namely a sort of integral over the infinite-dimensional space of paths satisfying the aforementioned boundary conditions. In
order to shed some light on the heuristics underpinning this formula,
let us briefly outline the so-called \emph{sequential approach} to
path integrals introduced by Nelson \cite{nelson}, which seems the
closest to Feynman's original formulation. First, recall that the
free propagator $e^{-itH_{0}}$ can be properly identified with a
Fourier multiplier and the following integral expression holds: 
\[
e^{-itH_{0}}\varphi\left(x\right)=\frac{1}{\left(2\pi it\right)^{d/2}}\int_{\mathbb{R}^{d}}\exp\left(i\frac{\left|x-x_0\right|^{2}}{2t}\right)\varphi(x_0)dx_0,\qquad\varphi\in\mathcal{S}(\mathbb{R}^{d}).
\]
Next, under suitable assumptions for the potential $V$, the Trotter product formula holds for the propagator generated by $H=H_{0}+V$:
\[
e^{-it\left(H_{0}+V\right)}=\lim_{n\rightarrow\infty}\left(e^{-i\frac{t}{n}H_{0}}e^{-i\frac{t}{n}V}\right)^{n},
\]
where the limit is intended in the strong topology of operators in
$L^{2}(\mathbb{R}^{d})$. Combining these two results gives
the following representation of the complete propagator $e^{-itH}$
as limit of integral operators (cf.\ \cite[Thm.\ X.66]{reed simon 2}):
\begin{equation}
e^{-it\left(H_{0}+V\right)}\phi(x)=\lim_{n\rightarrow\infty}\left(2\pi i\frac{t}{n}\right)^{-\frac{nd}{2}}\int_{\mathbb{R}^{nd}}e^{iS_{n}\left(t;x_{0},\ldots,x_{n-1},x\right)}\varphi\left(x_{0}\right)dx_{0}\ldots dx_{n-1},\label{prop ts limit}
\end{equation}
where
\[
S_{n}\left(t;x_{0},\ldots,x_{n-1},x\right)=\sum_{k=1}^{n}\frac{t}{n}\left[\frac{1}{2}\left(\frac{\left|x_{k}-x_{k-1}\right|}{t/n}\right)^{2}-V\left(x_{k}\right)\right], \quad x_{n}=x.
\]
In order to grasp the meaning of the phase $S_{n}\left(t;x_{0},\ldots,x_{n}\right)$,
consider the following argument: given the points $x_{0},\ldots,x_{n-1},x\in\mathbb{R}^{d}$,
let $\overline{\gamma}$ be
the polygonal path through the vertices $x_{k}=\overline{\gamma}\left(kt/n\right)$,
$k=0,\ldots,n$, $x_n=x$, parametrized as 
\[
\overline{\gamma}\left(\tau\right)=x_{k}+\frac{x_{k+1}-x_{k}}{t/n}\left(\tau-k\frac{t}{n}\right),\qquad\tau\in\left[k\frac{t}{n},\left(k+1\right)\frac{t}{n}\right],\qquad k=0,\ldots,n-1.
\]
Hence $\overline{\gamma}$ prescribes a classical motion with constant
velocity along each segment. The action for this path is thus given
by 
\[
S\left[\overline{\gamma}\right]=\sum_{k=1}^{n}\frac{1}{2}\frac{t}{n}\left(\frac{\left|x_{k}-x_{k-1}\right|}{t/n}\right)^{2}-\int_{0}^{t} V(\overline{\gamma}(\tau))d\tau.
\]
According to Feynman's interpretation, \eqref{prop ts limit} can
be thought as an integral over all polygonal paths and $S_{n}\left(x_{0},\ldots,x_{n},t\right)$
is a Riemann-like approximation of the action functional evaluated
on them. The limit $n\rightarrow\infty$ is now intuitively clear:
the set of polygonal paths becomes the set of all paths and we recover \eqref{prop path int}. In fact, it should be noted
that the custom in Physics community is to employ the suggestive formula
\eqref{prop path int} as a placeholder for \eqref{prop ts limit}
and the related arguments - see for instance \cite{gs,kleinert}.

We could not hope to frame the vast literature concerning the problem
of putting the formula \eqref{prop path int} on firm mathematical
ground; the interested reader could benefit from the monographs \cite{albeverio book,fujiwara5 book, mazzucchi} as points of departure. We only remark that the there is in general some relationship between the regularity assumptions on the potential and the strength
of the convergence of the time-slicing approximation. While the operator
theoretic strategy outlined above also allows to treat wild potentials,
the convergence in finer operator topologies (for instance, at the
level of integral kernels as in Feynman's original formulation) have been an open
problem for a long time. Nevertheless, there is a variety of schemes to deal
with path integrals and pointwise convergence of integral kernels
can be achieved by means of other sophisticated techniques, at least for smooth potentials - see the works of Fujiwara, Ichinose, Kumano-go and coauthors  
\cite{fujiwara1 fund sol,fujiwara2 duke,ichinose1,ichinose2,ichinose3,ichinose4,kumanogo0,kumanogo1,kumanogo2,kumanogo3,kumanogo4,kumanogo5,kumanogo6}. 

\subsection{Main results}\label{sec maint}
The present contribution aims at investigating the convergence at
the level of integral kernels for the time-slicing approximation of
path integrals under low regularity assumptions for the involved potentials.
We consider the Schr\"odinger
equation 
\begin{equation}\label{schr maint}
i\partial_{t}\psi=\left(H_{0}+V(x)\right)\psi,
\end{equation}
where now $H_{0}=a^{\text{w}}$ is the Weyl quantization of a real
quadratic form $a\left(x,\xi\right)$ on $\mathbb{R}^{2d}$ and
$V\in L^{\infty}(\mathbb{R}^{d})$ is complex-valued (so that a linear magnetic potential or a quadratic electric potential are allowed and included in $H_0$). It is well
known that the propagator $U_0(t)=e^{-itH_0}$ for the unperturbed problem $\left(V=0\right)$
is a metaplectic operator \cite{folland}. By a slight abuse of language (essentially, up to a sign factor), we can suggestively write $U_0(t)=\mu\left(\mathcal{A}_{t}\right)$, where
$t\mapsto\mathcal{A}_{t}\in\mathrm{Sp}\left(d,\mathbb{R}\right)$
is the one-parameter subgroup of symplectic matrices associated with
the solution of the classical equations of motion with Hamiltonian
$H_{0}$ in phase space  and $\mu$ is the so-called metaplectic representation - see Section \ref{weyl sec} for the rigorous construction of $U_0(t)$. We express the block structure of $\cA_t$, namely
\[
\mathcal{A}_{t}=\left(\begin{array}{cc}
A_{t} & B_{t}\\
C_{t} & D_{t}
\end{array}\right),
\]
since our results are global in time unless certain exceptional values, namely for any $t\in \bR$ such that $\det B_t \ne 0$ (equivalently, for any $t\in \bR$ such that $\cA_t$ is a free symplectic matrix - cf.\ Section \ref{weyl sec}). Consequently, we also introduce the quadratic form
\begin{equation}\label{phit}
\Phi_{t}\left(x,y\right)=\frac{1}{2}xD_{t}B_{t}^{-1}x-yB_{t}^{-1}x+\frac{1}{2}yB_{t}^{-1}A_{t}y,\qquad x,y\in\mathbb{R}^{d}.
\end{equation}
Recall that (cf.\ \cite{hormander1 mehler}) $H_{0}$ is a self-adjoint
operator on its domain 
\[
D\left(H_{0}\right)=\{ \psi \in L^{2}(\mathbb{R}^{d})\,:\,H_{0}\psi \in L^{2}(\mathbb{R}^{d})\} .
\]
Hence, $V$ being bounded, the Trotter product
formula holds: 
\begin{equation}\label{trotter maint}
e^{-it\left(H_{0}+V\right)}=\lim_{n\rightarrow\infty}E_{n}(t),\qquad E_{n}(t)=\Big(e^{-i\frac{t}{n}H_{0}}e^{-i\frac{t}{n}V}\Big)^{n},
\end{equation}
where the convergence is again in the strong operator topology in
$L^{2}(\mathbb{R}^{d})$  (see e.g. \cite[Cor.\ 2.7]{engel}). We denote by $e_{n,t}(x,y)$
the distribution kernel of $E_{n}(t)$ and by $u_t(x,y)$ that of $e^{-it\left(H_{0}+V\right)}$.

In order to state our first result, we need to introduce two spaces of a marked harmonic analysis flavour, defined in terms of the decay of the Fourier transform.
Let $M_{s}^{\infty}(\mathbb{R}^{d})$, $s\in\mathbb{R}$,
denote the subspace of temperate distributions $f\in\mathcal{S}'(\mathbb{R}^{d})$
such that, for some non-zero Schwartz function $g\in\mathcal{S}(\mathbb{R}^{d})$,
\[
\left\Vert f\right\Vert _{M_{s}^{\infty}}=\sup_{x,\xi \in\mathbb{R}^{d}}\left|\mathcal{F}\left[f\cdot g\left(\cdot-x\right)\right]\left(\xi\right)\right|\left(1+\left|\xi\right|\right)^{s}<\infty,
\]
where $\mathcal{F}$ is the Fourier transform. In addition, consider
the space $\mathcal{F}L_{s}^{1}(\mathbb{R}^{d})$ of functions
with weighted integrable Fourier transforms, namely:
\[
f\in\mathcal{F}L_{s}^{1}(\mathbb{R}^{d})\quad\Leftrightarrow\quad\left\Vert f\right\Vert _{\mathcal{F}L_{s}^{1}}=\int_{\mathbb{R}^{d}}\left|\mathcal{F}f\left(\xi\right)\right|\left(1+\left|\xi\right|\right)^{s}d\xi<\infty.
\]
While the space $\mathcal{F}L^{1}_s$ is rather standard, $M_{s}^{\infty}$
is a special member of a family of Banach spaces, the so-called \textit{modulation spaces}, stemming from the
branch of harmonic analysis currently known as time-frequency analysis
(cf.\ Section \ref{sec mod sp} for the details). Modulation spaces proved to
be a fruitful environment for the study of PDEs, in particular the Schr\"odinger equation (see for instance \cite{CGNR jmp,CNR rough,CN pot mod,wang} and the references therein), and related problems such
as path integrals \cite{nicola1 conv lp,nicola2 ks,NT}. We have a convergence result for potentials in this space.
\begin{theorem}
	\label{maint minfty}Consider  $H_{0}$
	as specified above and $V\in M_{s}^{\infty}(\mathbb{R}^{d})$,
	with $s>2d$. For any $t\in\mathbb{R}$ such that $\det B_{t}\neq0$:
	\begin{enumerate}
		\item the distributions $e^{-2\pi i\Phi_{t}}e_{n,t}$, $n\geq 1$,
		and $e^{-2\pi i\Phi_{t}}u_t$
		belong to a bounded subset of $M_{s}^{\infty}(\mathbb{R}^{2d})$;
		\item $e_{n,t}\rightarrow u_t$ in $\left(\mathcal{F}L_{r}^{1}\right)_{\mathrm{loc}}(\mathbb{R}^{2d})$
		for any $0<r<s-2d$, hence uniformly on compact subsets. 
	\end{enumerate}
\end{theorem}
We notice that for $s>2d$ we have $M_{s}^{\infty}(\mathbb{R}^{2d})\subset (\mathcal{F} L^1)_{\rm loc}(\rdd) \cap L^\infty(\R^{2d})$, so that the kernels $e_{n,t}$ and $u_t$ in the statement are in fact bounded and continuous functions, provided $\det B_{t}\neq0$. \par
Also,  $\cap_{s>0} M_{s}^{\infty}(\mathbb{R}^{d}) = C^{\infty}_b (\mathbb{R}^{d})$ is the space of bounded smooth functions with bounded derivatives of any order, which gives the following result. 
\begin{corollary}
	\label{maint s000} Let $H_{0}$ be
	as specified above and $V\in C^{\infty}_b(\mathbb{R}^{d})$. For any
	$t\in\mathbb{R}$ such that $\det B_{t}\neq0$:
	\begin{enumerate}
		\item the distributions $e^{-2\pi i\Phi_{t}}e_{n,t}$, $n\geq 1$,
		and $e^{-2\pi i\Phi_{t}}u_t$
		belong to a bounded subset of $C^{\infty}_b(\mathbb{R}^{2d})$;
		\item $e_{n,t}\rightarrow u_t$ in $C^{\infty}(\rdd)$,
		hence uniformly on compact  subsets together with any derivatives. 
	\end{enumerate}
\end{corollary}
The same conclusion of Corollary \ref{maint s000} is actually known
to hold true for short times, as a consequence of Fujiwara's result \cite{fujiwara2 duke}, but the above result is global in time. The occurrence of a set of exceptional times is to be expected: in these cases, the integral kernel of the propagator degenerates into a distribution. A basic example of this behaviour is given by the harmonic oscillator, that is\footnote{See Section \ref{prelres} for the choice of the normalization of the Weyl quantization and the classical flow.} $H_0=-\frac{1}{4\pi}\triangle +\pi |x|^2$, $V(x)=0$, at $t=k\pi$, $k\in\mathbb{Z}$. Notice that such exceptional values are exactly those for which the upper-right block of the associated Hamiltonian flow
\[ 
\cA_t = \left( \begin{array}{cc}
(\cos t) I & (\sin t)I \\ 
-(\sin t)I & (\cos t)I
\end{array} \right)
\] 
is non-invertible. 

We now state our main result, which is subtler than Theorem \ref{maint minfty} and applies to potentials in a lower regularity
space known as the Sj\"ostrand class $M^{\infty,1}(\mathbb{R}^{d})$:
we say that $f\in\mathcal{S}'(\mathbb{R}^{d})$ belongs
to $M^{\infty,1}(\mathbb{R}^{d})$ if \[
\left\Vert f\right\Vert _{M^{\infty,1}}=\int_{\mathbb{R}^{d}}\sup_{x\in\mathbb{R}^{d}}\left|\mathcal{F}\left[f\cdot g\left(\cdot-x\right)\right]\left(\xi\right)\right|d\xi<\infty,
\] for some non-zero $g\in\mathcal{S}(\mathbb{R}^{d})$. As a rule of thumb, a function in $M^{\infty,1}(\rd)$ is bounded on $\rd$ and locally enjoys the mild regularity of the Fourier transform of an $L^1$ function; in fact 
\[
(M^{\infty,1})_{\rm loc}(\rd) = (\mathcal{F}L^1)_{\rm loc}(\rd).
\]
Furthermore, we have the following chain of strict inclusions for $s>d$:
\[ C^{\infty}_b(\rd)\subset M^{\infty}_s(\rd) \subset M^{\infty,1}(\rd) \subset (\mathcal{F}L^1)_{\rm loc}(\rd)\cap L^\infty(\rd)\subset C^0(\rd)\cap L^\infty(\rd).
\] 
Intuitively: we have a scale of low-regularity spaces, the functions in $M^{\infty}_s(\rd)$ becoming less regular as $s\searrow d$, until the (fractional) differentiability is completely lost in the ``maximal" space $M^{\infty,1}(\rd)$. 

 It seems worth to highlight that results on the convergence of path integrals are already known for special
elements of the Sj\"ostrand class: for instance, a class of potentials widely
investigated by means of different approaches in the papers of Albeverio and coauthors \cite{albeverio1 inv,albeverio2 trace,albeverio3 jfa} and It\^o \cite{ito2} (see also \cite{ito1}) is $\cF\cM(\rd)$, namely the space of Fourier transforms of (finite) complex measures on $\rd$. In fact, we have $\cF\cM(\rd) \subset M^{\infty,1}(\rd)$, cf.\ Proposition \ref{fou meas sjo} below, and the above inclusion is strict; for instance, $f(x)=\cos |x|$, $x\in \rd$, clearly belongs to $C^{\infty}_b(\rd)$, but it is easy to realize that $f\notin \cF\cM(\rd)$ as soon as $d>1$, by the known formula for the fundamental solution of the wave equation \cite{evans}.

The following result encompasses these potentials and ultimately yields
the desired pointwise convergence \textit{at the level of integral kernels} for a wide class of non-smooth potentials.
\begin{theorem}
	\label{maint sjo} Let $H_{0}$ be
	as specified above and $V\in M^{\infty,1}(\mathbb{R}^{d})$.
	For any $t\in\mathbb{R}$ such that $\det B_{t}\neq0$:
	\begin{enumerate}
		\item the distributions $e^{-2\pi i\Phi_{t}}e_{n,t}$, $n\geq 1$,
		and $e^{-2\pi i\Phi_{t}}u_t$
		belong to a bounded subset of $M^{\infty,1}(\mathbb{R}^{2d})$;
		\item $e_{n,t}\rightarrow u_t$ in $\left(\mathcal{F}L^{1}\right)_{\mathrm{loc}}(\rdd)$,
		hence uniformly on compact subsets. 
	\end{enumerate}
\end{theorem}
Let us conclude this introduction with a few words on the techniques
employed for the proofs. The main idea is to rephrase the problem
in terms of pseudodifferential calculus and then to exploit the very rich structure
enjoyed by the modulation spaces $M_{s}^{\infty}(\mathbb{R}^{2d})$
(with $s>2d$) and $M^{\infty,1}(\mathbb{R}^{2d})$: in
particular, they are Banach algebras for both pointwise multiplication
and twisted product of symbols for the Weyl quantization - see the
subsequent Section \ref{weyl sec} for the details.  \par
There is a certain number of questions which seem worthy of further consideration. For example, Theorem \ref{maint minfty} and Corollary \ref{maint s000} should hopefully extend to Hamiltonians $H_0$ given by the Weyl quantization of a smooth real-valued function with derivatives of order $\geq 2$ bounded, using techniques from \cite{nicola1 conv lp}. Also, the potential $V$ could be replaced by a genuine pseudodifferential operator in suitable classes. We preferred to avoid further technicalities here, since the arguments below are already somewhat involved. Finally we observe that the techniques introduced in the present paper could hopefully be useful to study similar convergence problems of the integral kernels for other approximation formulas arising in semigroup theory; cf.\ \cite{engel}.\par\medskip
The paper is organized as follows. Sections 2 and 3 are both devoted to preliminary results and technical lemmas on function spaces and operators involved. In Section 4 we prove Theorem \ref{maint minfty} and Corollary \ref{maint s000}. Theorem \ref{maint sjo} is proved in Section 5. 

\section{Preliminary results}\label{prelres}
\subsection{Notation} We define $x^2=x\cdot x$, for $x\in\mathbb{R}^d$, and
$xy=x\cdot y$ is the scalar product on $\mathbb{R}^{d}$. The Schwartz class is denoted by  $\mathcal{S}(\mathbb{R}^{d})$, the space of temperate distributions by  $\mathcal{S}'(\mathbb{R}^{d})$. The brackets  $\langle  f,g\rangle $ denote the extension to $\cS' (\mathbb{R}^{d})\times \cS (\mathbb{R}^{d})$ of the inner product $\langle f,g\rangle=\int_{\rd} f(x){\overline {g(x)}}dx$ on $L^2(\mathbb{R}^{d})$. 

The conjugate exponent $p'$ of $p \in [1,\infty]$ is defined by $1/p+1/p'=1$. The symbol $\lesssim$ means that the underlying inequality holds up to a positive constant factor $C>0$.
For any $x\in\mathbb{R}^{d}$ and $s\in\mathbb{R}$
we set $\left\langle x\right\rangle ^{s}\coloneqq\left(1+\left|x\right|^{2}\right)^{s/2}$.
We choose the following normalization for the Fourier transform:
\[
\mathcal{F}f\left(\xi\right)=\int_{\mathbb{R}^{d}}e^{-2\pi ix\xi}f(x) d x,\qquad\xi\in \rd. \]
We define the translation and modulation operators: for
any $x,\xi \in\mathbb{R}^{d}$ and $f\in\mathcal{S}(\mathbb{R}^{d})$,
\[
\left(T_{x}f\right)\left(y\right)\coloneqq f(y-x),\qquad\left(M_{\xi}f\right)(y)\coloneqq e^{2 \pi i \xi y}f(y).
\]
These operators can be extended by duality on temperate distributions.
The composition $\pi(x,\xi)=M_\xi T_x$ constitutes a so-called time-frequency shift.

Denote by $J$ the canonical symplectic matrix in $\mathbb{R}^{2d}$:
\[
J=\left(\begin{array}{cc}
0_d & I_d\\
-I_d & 0_d
\end{array}\right)\in\mathrm{Sp}(d,\mathbb{R}),
\]
where  the
symplectic group $\mathrm{Sp}\left(d,\mathbb{R}\right)$ is defined
as:
\[ \mathrm{Sp}\left(d,\mathbb{R}\right)=\left\{M\in \mathrm{GL}(2d,\bR) \,:\;M^{\top}JM=J\right\}
\]
and the associated Lie algebra is
\[ \mathfrak{sp}(d,\bR) \coloneqq \{ M \in \bR^{2d\times 2d}\,:\, MJ+JM^{\top}=0\}. \]

\subsection{Modulation spaces} \label{sec mod sp}  
The short-time Fourier transform (STFT) of a temperate distribution $f\in\cS'(\rd)$ with respect to the window function $g \in \cS(\rd)\setminus\{0\}$ is defined by:
\begin{equation}\label{STFTdef}
V_gf (x,\xi)\coloneqq \langle f, \pi(x,\xi) g\rangle=\Fur (f\cdot T_x g)(\xi)=\int_{\rd}e^{-2\pi iy \xi }
f(y)\, {\overline {g(y-x)}}\,dy.
\end{equation}

This is a key instrument for time-frequency analysis; the monograph \cite{gro1 book} contains a comprehensive treatment of its mathematical properties, especially those mentioned below. For the sake of conciseness, we only mention that the STFT is deeply connected with other well-known phase-space transforms, in particular the Wigner transform
$$W(f,g)(x,\xi )=\int_{\mathbb{R}^{d}}e^{-2\pi iy \xi }f\left(x+\frac{y}{2}\right)%
\overline{g\left(x-\frac{y}{2}\right)}\ dy. $$
For this and other aspects of the connection with phase space analysis, we recommend \cite{dG symp met}. 

Given a non-zero window $g\in\cS(\rd)$, $s\in \bR$ and $1\leq p,q\leq
\infty$, the {\it
	modulation space} $M^{p,q}_s(\rd)$ consists of all temperate
distributions $f\in\sch'(\rd)$ such that $V_gf\in L^{p,q}_s(\rdd)$ (mixed weighted Lebesgue space), that is: 
$$
\|f\|_{M^{p,q}_s}=\|V_gf\|_{L^{p,q}_s}=\left(\int_{\rd}
\left(\int_{\rd}|V_gf(x,\xi)|^p\,
dx\right)^{q/p} \la \xi \ra^{qs} d\xi \right)^{1/q}  \, <\infty ,
$$
with trivial adjustments if $p$ or $q$ is $\infty$. 
If $p=q$, we write $M^p$ instead of $M^{p,p}$. For the unweighted case, corresponding to $s=0$, we omit the dependence
on $s$: $M_{0}^{p,q}\equiv M^{p,q}$. 

It can be proved that $M^{p,q}_s(\rd)$ is a Banach space whose definition does not depend on the choice of the window $g$. Just to get acquainted with this family, it is worth to mention that many common function spaces can be equivalently designed as modulation spaces: for instance,  
\begin{enumerate}[label=(\roman*)]
	\item $M^2(\rd)$ coincides with the Hilbert space $L^2(\rd)$;
	\item $M^2_s(\rd)$ coincides with the usual $L^2$-based Sobolev space $H^s(\rd)$;
	\item the following continuous embeddings with Lebesgue spaces hold:
	\[ M^{p,q}_r(\rd) \hookrightarrow L^p(\rd) \hookrightarrow M^{p,q}_s(\rd), \qquad r>d/q' \text{ and } s<-d/q. \] In particular, 
	\[ M^{p,1}(\rd) \hookrightarrow L^p(\rd) \hookrightarrow M^{p,\infty}(\rd).\] 
\end{enumerate}
For these and other properties we address the reader to \cite{fei new segal,fei modulation 83,gro1 book}. 

For a fixed window $g \in \cS(\rd)\setminus\{0\}$, the STFT operator $V_g$ is clearly bounded from $M^{p,q}_s(\rd)$ to $L^{p,q}_s(\rdd)$. The \textit{adjoint operator} of $V_g$, defined by \[ V_g^* F = \int_{\rdd} F(x,\xi) \pi(x,\xi) g \, dxd\xi, \] continuously maps the Banach space $L^{p,q}_s(\rdd)$ into $M^{p,q}_s(\rd)$, the integral above to be intended in a weak sense.\par
The \textit{inversion formula for the STFT} can be conveniently expressed as follows: for any $f \in M^{p,q}_s(\rd)$,
\begin{equation} \label{inversion STFT}
f = \frac{1}{\| g \|_{L^2}^2} V_g^* V_g f,
\end{equation}
again in a weak sense. \par
The Sj\"ostrand's class, originally defined in \cite{sjo}, coincides with the choice $p=\infty$, $q=1$, $s=0$. We have that $M^{\infty,1}(\rd)\subset C^0(\rd) \cap L^{\infty}(\rd)$ and it is a Banach algebra under pointwise product. In fact, precise conditions are known on $p$, $q$ and $s$ in order for $M_{s}^{p,q}$
to be a Banach algebra with respect to pointwise multiplication:
\begin{lemma}[{\cite[Thm.\ 3.5]{rs mod}}] \label{Mpqs ban alg}Let $1\le p,q\le\infty$ and $s\in\mathbb{R}$. The modulation space $M_{s}^{p,q}(\mathbb{R}^{d})$ is
	a Banach algebra for pointwise multiplication if and only if either
	$s=0$ and $q=1$ or $s>d/q'$. 
\end{lemma}
Therefore, the Sj\"ostrand's class $M^{\infty,1}(\mathbb{R}^{d})$
and the modulation spaces $M_{s}^{\infty}(\mathbb{R}^{d})$
with $s>d$ are Banach algebras for pointwise multiplication. It is
worth to point out that the condition required in Lemma \ref{Mpqs ban alg}
are in fact equivalent to assume $M_{s}^{p,q}\hookrightarrow L^{\infty}$
- cf.\ \cite[Cor.\ 2.2]{rs mod}. 

\begin{remark} We clarify once for all that the preceding results
concern the conditions under which the embedding $M_{s}^{p,q}\cdot M_{s}^{p,q}\hookrightarrow M_{s}^{p,q}$
is continuous, hence there exists a constant $C>0$ such that
\[
\left\Vert fg\right\Vert _{M_{s}^{p,q}}\le C\left\Vert f\right\Vert _{M_{s}^{p,q}}\left\Vert g\right\Vert _{M_{s}^{p,q}},\qquad\forall f,g\in M_{s}^{p,q}.
\]
Thus, the algebra property holds up to a constant. It is a well known
general fact that one can provide an equivalent
norm for which the previous estimate holds with $C=1$ and the unit element of the algebra has norm equal to $1$ (cf.\ \cite[Thm.\ 10.2]{rudin fa}). From
now on, we assume to work with such equivalent norm whenever concerned
with a Banach algebra. 
\end{remark}

An important subspace of both $M^{\infty,1}(\rd)$ and $M^{\infty}_s(\rd)$ is the space
\[
C^{\infty}_b(\mathbb{R}^{2d})\coloneqq\left\{ f\in C^{\infty}(\mathbb{R}^{2d})\,:\,\left|\partial^{\alpha}f\right|\leq C_{\alpha}\quad\forall\alpha\in\mathbb{N}^{2d}\right\}=\bigcap_{s\ge0}M_{s}^{\infty}(\mathbb{R}^{d});\]
see e.g. \cite[Lem.\ 6.1]{gro3 rze} for this characterization.

 We briefly mention that the image of modulation spaces under Fourier transform yields another important family of function spaces for the purposes of real harmonic analysis, which are a very special type of \textit{Wiener amalgam spaces}: for any $1\le p,q \le \infty$, we set 
\[
  W^{p,q}(\rd)  \coloneqq \cF M^{p,q}(\rd). 
\]
One can prove that $W^{p,q}(\rd)$ is a Banach space under the same norm of $M^{p,q}(\rd)$ but with flipped order of integration with respect to the time and frequency variables:  \[\|f\|_{W^{p,q}} \coloneqq \left(\int_{\rd}
\left(\int_{\rd}|V_gf(x,\xi)|^p (\xi)\,
d\xi \right)^{q/p} d x\right)^{1/q},
\]
for $g\in \cS(\rd)\setminus \{0\}$, as usual.
  
\subsection{Weyl operators} \label{weyl sec}
The usual definition of the Weyl transform of the symbol $\sigma:\mathbb{R}^{2d}\rightarrow\mathbb{C}$
is
\[
\sigma^{\text{w}}f\left(x\right)\coloneqq\int_{\mathbb{R}^{2d}}e^{2\pi i\left(x-y\right)\xi}\sigma\left(\frac{x+y}{2},\xi\right)f\left(y\right)dyd\xi.
\]
The meaning of this formal integral operator heavily relies on the
function space to which the symbol $\sigma$ belongs. Instead, we
adopt the following definition via duality for symbols $\sigma\in\mathcal{S}'(\mathbb{R}^{2d})$:
\begin{equation}
\sigma^{\mathrm{w}}:\mathcal{S}(\mathbb{R}^{d})\rightarrow\mathcal{S}'(\mathbb{R}^{d}),\quad\quad\langle \sigma^{\mathrm{w}}f,g\rangle =\langle \sigma,W(g,f)\rangle ,\qquad\forall f,g\in\mathcal{S}(\mathbb{R}^{d}).\label{def wig dual}
\end{equation}
In particular, $M^{\infty,1}(\mathbb{R}^{2d})$ and $M_{s}^{\infty}(\mathbb{R}^{2d})$
are suitable symbol classes. It is worth to mention that the classical
symbol classes investigated within the long tradition of pseudodifferential
calculus are usually defined by means of decay/smoothness conditions
(see for instance the general H\"ormander classes $S_{\rho,\delta}^{m}(\mathbb{R}^{2d})$ - \cite{hormander2 book 3}),
while the fruitful interplay with time-frequency analysis allows to
cover very rough symbols too - cf.\ \cite{gro ped}.
 
\begin{remark}
	\label{mult symbol}Notice that the multiplication by $V(x)$
	is a special example of Weyl operator with symbol \[ \sigma_V \left(x,\xi\right)=V\left(x\right)=(V\otimes 1)(x,\xi), \quad (x,\xi)\in \rdd.\] 
	It is not difficult to prove that the correspondence $V\mapsto \sigma_V$ is continuous from $M^{\infty}_s(\rd)$ (resp. $M^{\infty,1}(\rd)$) to $M^{\infty}_s(\rdd)$ (resp. $M^{\infty,1}(\rdd)$). 
	In the rest of the paper this identification shall be implicitly assumed; by a slight abuse of notation, we will write $V$ also for $\sigma_V^{\mathrm{w}}$ for the sake of legibility.  
\end{remark} 
\noindent
The composition of Weyl transforms provides a bilinear form on symbols, the so-called \emph{twisted product}:
\[ \sigma^{\mathrm{w}} \circ \rho^{\mathrm{w}} = (\sigma \# \rho)^{\mathrm{w}}. \]

Although explicit formulas for the twisted product of symbols can be derived (cf.\ \cite{wong}), we will not need them hereafter. Anyway, this is a fundamental notion in order to establish an algebra structure on symbol spaces. It is a distinctive property of $M^{\infty,1}(\mathbb{R}^{2d})$, as well as $M_{s}^{\infty}(\mathbb{R}^{2d})$ with $s>2d$,
to enjoy a double Banach algebra structure: \begin{itemize}
	\item a commutative one with respect to the pointwise multiplication as detailed above;
	\item a non-commutative one with respect to the twisted product of symbols (\cite{gro3 rze,sjo}). 
\end{itemize} 

We wish to underline that the latter algebra structure has been deeply investigated from a time-frequency analysis perspective. Indeed, it is subtly related to a characterizing property satisfied by pseudodifferential operators with symbols in those spaces, namely \textit{almost diagonalization} with respect to time-frequency shifts: we have $\sigma\in M^{\infty}_s(\rdd)$ if and only if
\[ \lvert \la \sigma^{\mathrm{w}} \pi(z) \varphi, \pi(w)\varphi \ra \rvert \le C\la w-z \ra^{-s}, \quad \forall \varphi \in \cS(\rd)\setminus\{0\},\ z,w\in \rd. \]
Similarly, $\sigma\in M^{\infty,1}(\rdd) $ if and only if there exists $H\in L^1(\rdd)$ such that
\[ \lvert \la \sigma^{\mathrm{w}} \pi(z) \varphi, \pi(w)\varphi \ra \rvert \le H(w-z), \quad \forall \varphi \in \cS(\rd)\setminus\{0\},\ z,w\in \rd. \]
We address the reader to \cite{CGNR fio,CGNR jmp,CNT 18,gro2 sj,gro3 rze} for further discussions on these aspects.  
\begin{remark}\label{rem algebra}
To unambiguously fix the notation: whenever concerned with a product of elements $a_1,\ldots,a_N$ in a Banach algebra $(A,\star)$, we write \[
\prod_{k=1}^N a_k \coloneqq a_1 \star a_2 \star \ldots \star a_N.
\]
This relation is meant to hold even when $(A,\star)$ is a non-commutative algebra, provided that the symbol on the LHS exactly designates the ordered product on the RHS.  
\end{remark}

\subsection{Metaplectic operators}\label{mataop} Given a symplectic matrix $\cA \in \Spdr$, we say that the unitary operator $\mu(\cA) \in \cU(L^2(\rd))$ is a \textit{metaplectic operator} associated with $\cA$ if it does satisfy the following intertwining relation:
\[
\pi(\cA z)=\mu(\cA) \pi(z) \mu(\cA)^{-1},\quad \forall z \in \rdd.
\]
Strictly speaking, the previous formula defines a whole set of unitary operators up to a constant phase factor: $ \{ c_{\cA}\mu(\cA) : c_{\cA} \in \bC, \lvert c_{\cA} \rvert = 1\}$. The phase factor can be adjusted to either $c_{\cA}=1$ or $c_{\cA}=-1$, namely:
\[ \mu(AB) = \pm \mu(A)\mu(B), \quad \forall A,B\in \Spdr. \] That is, $\mu$ provides a double-valued unitary representation of $\Spdr$ or, better, a representation of the double covering $\Mpdr$ of $\Spdr$; we will denote by $\rho^{\mathrm{Mp}}:  \Mpdr\to\Spdr$ the projection. We refer to \cite{dG symp met,folland} for a comprehensive discussion of these aspects. 

We confine ourselves to recall that the metaplectic operator corresponding to special symplectic matrices can be explicitly written as a quadratic Fourier transform. We say that $\cA \in \Spdr$, with \[ \cA= \left( \begin{array}{cc}
A & B \\ C & D \end{array} \right), \] is a \textit{free symplectic matrix} whenever $\det B \neq 0$. We have the following integral formula\footnote{We underline that the following quadratic Fourier transform, up to a sign in $\Phi_{\cA}$, is actually the point of departure for the construction of the metaplectic representation in \cite{dG symp met}.} for $\mu(\cA)$.

\begin{theorem}[{\cite[Sec.\  7.2.2]{dG symp met}}] Let $\cA \in \Spdr$ be a free symplectic matrix. Then,
	\begin{equation} \label{met int formula} \mu(\cA)f(x) = c \lvert \det B\rvert^{-1/2} \int_{\rd} e^{2\pi i \Phi_{\cA}(x,\xi)} f(y) dy, \qquad f \in \cS(\rd), \end{equation} where $c\in \bC$ is a suitable complex factor of modulus $1$ and $\Phi_{\cA}$ is the quadratic form given by
\[ \Phi_{\cA}(x,y) = \frac{1}{2} xDB^{-1}x - y B^{-1}x + \frac{1}{2}y B^{-1}A y. \]
\end{theorem}
 Incidentally, notice that $\mu(J) = c \cF^{-1}$. 

It is important to recall that a truly distinctive property of the Weyl quantization is its \textit{symplectic covariance} {\cite[Thm.\ 215]{dG symp met}}, namely: for any $\cA \in \Spdr$ and $\sigma \in \cS'(\rdd)$, the following relation holds:
\begin{equation}\label{symp cov}
(\sigma \circ \cA)^{\mathrm{w}} = \mu(\cA)^{-1}\sigma^{\mathrm{w}} \mu(\cA).
\end{equation} 

Let now $a$ be a real-valued, time-independent, quadratic homogeneous polynomial on $\rdd$, namely: 
\[ a(x,\xi) = \frac{1}{2}xAx + \xi B x + \frac{1}{2}\xi C \xi, \] where $A,C\in \bR^{d\times d}$ are symmetric matrices and $B\in \bR^{d\times d}$. The phase-space flow determined by the Hamilton equations\footnote{The factor $2\pi$ derives from the normalization of the Fourier transform adopted in the paper.}
\[ 2 \pi \dot{z} = J \nabla_z a(z) = \cA, \quad \cA= \left(\begin{array}{cc} B & C \\ -A & -B^{\top}\end{array}\right) \in \mathfrak{sp}(d,\bR), \]
defines a mapping $\bR \ni t  \mapsto \cA_t = e^{(t/2\pi)\cA} \in \Spdr. $ It follows from the general theory of covering manifolds that this path can be lifted in a unique way to a mapping 
$ \bR \ni t \mapsto M_0(t) \in \Mpdr, \ M_0(0) = I$; hence $ \rho^{\mathrm{Mp}}(M_0(t)) = \cA_t$. Then, with $H_0 = a^{\mathrm{w}}$, the Schr\"odinger equation
\[\begin{cases}
i\partial_{t}\psi=H_0 \psi\\
\psi\left(0,x\right)=\varphi\left(x\right), \quad \varphi \in \cS(\rd),
\end{cases} \] is solved by  
\[ \psi(t,x)=e^{-it H_0}\varphi(x)=\mu(M_0(t))\varphi(x),
\]
see \cite[Sec.\ 15.1.3]{dG symp met}. By a slight abuse of language we will write $\mu(\cA_t)$ in place of $\mu(M_0(t))$. 
We recommend \cite{CN pot mod,dG symp met, folland} for further details on the matter.
\subsection{Operators and kernels}
Consider the space $\cL(X,Y)$ of all continuous linear mappings between two Hausdorff topological vector spaces $X$ and $Y$. It can be endowed with different topologies \cite{treves}, in which cases we write:
\begin{enumerate} 
	\item $\cL_b(X,Y)$, if equipped with the topology of \textit{bounded convergence}, that is uniform convergence on bounded subsets of $X$;
	\item $\cL_c(X,Y)$, if equipped with the topology of \textit{compact convergence}, that is uniform convergence on compact subsets of $X$;
	\item $\cL_s(X,Y)$, if equipped with the topology of \textit{pointwise convergence}, that is uniform convergence on finite subsets of $X$.
\end{enumerate}
Notice that if $Y=\bC$, $\cL_b(X,Y) = X_b'$ (the strong dual of $X$), while $\cL_s(X,Y)= X_s'$ (the weak dual of $X$). 
We will be mainly concerned with the case $X=\cS(\rd)$ and $Y=\cS'(\rd)$, the latter always endowed with the strong topology unless otherwise specified (i.e., $\cS'(\rd)=\cL_b(\cS(\rd),\bC)$). 
The celebrated Schwartz kernel theorem is usually invoked for proving that any reasonably well-behaved operator is indeed an integral transform, though in a distributional sense. In the following we will need this identification but at the topological level \cite{gask,treves}, that is, a linear map $A:\cS(\rd)\rightarrow \cS'(\rd)$ is continuous if and only if it is generated by a  (unique) temperate distribution $K\in \cS'(\rdd)$, namely:
	\begin{equation}
	\langle Af,g\rangle =\langle K,g\otimes\overline{f}\rangle ,\qquad\forall f,g\in\mathcal{S}(\mathbb{R}^{d}),\label{kernel eq}
	\end{equation}
and the correspondence $K\mapsto A$ above is a topological isomorphism between $\mathcal{S}'(\mathbb{R}^{2d})$ and the space $\mathcal{L}_b\left( \cS(\rd), \cS'(\rd)\right)$.
As mentioned above, $\cS'(\rd)$ and $\cS'(\rdd)$ are endowed with the strong topology.
\begin{proposition}\label{propokernel}
Let $A_n\to A$ in $\mathcal{L}_s\left( \cS(\rd), \cS'(\rd)\right)$. Then we have convergence in $\mathcal{S}'(\mathbb{R}^{2d})$ of the corresponding distribution kernels. 
\end{proposition}
\begin{proof}
Since $\cS(\rd)$ is a Fréchet space and $A_n$, being a sequence, defines a filter with countable basis on $\cL_s(\cS(\rd),\cS'(\rd))$, from \cite[Cor.\ at pag. 348]{treves} we have that $A_n\rightarrow A$ also in $\cL_c(\cS(\rd),\cS'(\rd))$, which is in turn equivalent to convergence in $\cL_b(\cS(\rd),\cS'(\rd))$ since $\cS(\rd)$ is a Montel space - cf.\ \cite[Prop.\ 34.4 and 34.5]{treves}. The desired conclusion then follows from Schwartz' kernel theorem.
\end{proof}

\section{Technical lemmas}
The following lemma extends \cite[Lem.\ 2.2 and Prop.\ 5.2]{CGNR jmp}.
\begin{lemma}\label{tech lem}
Let $X$ denote either $M^{\infty}_s(\rdd)$, $s\ge 0$, or $M^{\infty,1}(\rdd)$.  \begin{enumerate}[label=(\roman*)]
	
	\item Let $\sigma\in X$
	and $t\mapsto\mathcal{A}_{t}\in\mathrm{Sp}\left(d,\mathbb{R}\right)$
	be a continuous mapping defined on the compact interval $\left[-T,T\right]\subset\mathbb{R}$,
	$T>0$. For any $t\in\left[-T,T\right]$, we have $\mathcal{\sigma}\circ\mathcal{A}_{t}\in X$,
	with 
	\begin{equation}\label{estimate linear comp}
	\left\Vert \sigma\circ\mathcal{A}_{t}\right\Vert _{X}\le C\left(T\right)\|\sigma\|_X.
	\end{equation}
	\item Let $\sigma \in X$ and $A,B,C$ be real $d\times d$ matrices with $B$ invertible, and set $$\Phi(x,y)=\frac{1}{2}xAx + yBx + \frac{1}{2}yCy. $$ There exists a unique symbol $\widetilde{\sigma} \in X$ such that, for any $f\in \cS(\rd)$:
	\begin{equation}\label{swap weyl fio}
	\sigma^{\mathrm{w}} \int_{\rd} e^{2\pi i \Phi(x,y)}f(y)dy = \int_{\rd} e^{2\pi i \Phi(x,y)}\widetilde{\sigma}(x,y)f(y)dy.
	\end{equation}
	Furthermore, the map $\sigma \mapsto \widetilde{\sigma}$ is bounded on $X$.

	\end{enumerate}
\end{lemma}
\begin{proof} The case $X=M^{\infty,1}(\rdd)$ is covered by \cite[Lem.\ 2.2]{CGNR jmp}. We prove here the claim for $X=M^{\infty}_s(\rdd)$.\par
$(i)$	For any non-zero
	window function $\Phi\in\mathcal{S}\left(\mathbb{R}^{2d}\right)$ and $\mathcal{A}\in\mathrm{Sp}\left(d,\mathbb{R}\right)$
	we have
	\begin{align*}
	\left\Vert \sigma\circ\mathcal{A}\right\Vert _{M_{s}^{\infty}} & =\sup_{z,\zeta\in\mathbb{R}^{2d}}\left|\left\langle \sigma\circ\mathcal{A},M_{\zeta}T_{z}\Phi\right\rangle \right|\left\langle \zeta\right\rangle ^{s}\\
	& =\sup_{z,\zeta\in\mathbb{R}^{2d}}\left|\left\langle \sigma,M_{\left(\mathcal{A}^{-1}\right)^{\top}\zeta}T_{\mathcal{A}z}\left(\Phi\circ\mathcal{A}^{-1}\right)\right\rangle \right|\left\langle \zeta\right\rangle ^{s}\\
	& =\sup_{z,\zeta\in\mathbb{R}^{2d}}\left|\left\langle \sigma,M_{\zeta}T_{z}\left(\Phi\circ\mathcal{A}^{-1}\right)\right\rangle \right|\left\langle \mathcal{A}^{\top}\zeta\right\rangle ^{s}\\
	& \leq\left\Vert \mathcal{A}^{\top}\right\Vert ^{s}\left\Vert V_{\Phi\circ\mathcal{A}^{-1}}\sigma\right\Vert _{M_{s}^{\infty}}\\
	& \lesssim \left\Vert \mathcal{A}\right\Vert ^{s}\left\Vert V_{\Phi\circ\mathcal{A}^{-1}}\Phi\right\Vert _{L_{s}^{1}}\left\Vert \sigma\right\Vert _{M_{s}^{\infty}},
	\end{align*}
	where we used the estimate $\left\langle \mathcal{A}^{\top}\zeta\right\rangle ^{s}\le\left\Vert \mathcal{A}^{\top}\right\Vert ^{s}\left\langle \zeta\right\rangle ^{s}$
	(here $\left\Vert \mathcal{B}\right\Vert $ denotes the operator norm
	of the matrix $\mathcal{B}$) and the change-of-window Lemma \cite[Lem.\ 11.3.3]{gro1 book} ($\|\cdot\|_{L_{s}^{1}}$ denoting the weighted $L^1$ norm with weight $\langle\zeta\rangle^{s}$) . 
		
	We now prove the uniformity with respect to the parameter $t$, when $\mathcal{A}=\mathcal{A}_{t}$. The
	subset $\left\{ \mathcal{A}_{t}\,:\,t\in\left[-T,T\right]\right\} \subset\mathrm{Sp}\left(d,\mathbb{R}\right)$
	is bounded and thus $\left\Vert \mathcal{A}_{t}\right\Vert \le C_{1}\left(T\right)$.
	Furthermore, $\left\{ V_{\Phi\circ\mathcal{A}_{t}^{-1}}\Phi\,:\,t\in\left[-T,T\right]\right\} $
	is a bounded subset of $\mathcal{S}(\mathbb{R}^{2d})$
	(this follows at once by inspecting the Schwartz seminorms of
	$\Phi\circ\mathcal{A}_{t}^{-1}$), hence $\left\Vert V_{\Phi\circ\mathcal{A}^{-1}}\Phi\right\Vert _{L_{s}^{1}}\le C_{2}\left(T\right)$. 
	
	$(ii)$ The proof is similar to that of the case $X=M^{\infty,1}(\rdd)$ in \cite[Prop.\ 5.2]{CGNR jmp}. In particular, $\widetilde{\sigma}$ is explicitly derived from $\sigma$ as follows: $\widetilde{\sigma} = \cU_2 \cU \cU_1 \sigma$, where $\cU,\cU_1,\cU_2$ are the mappings
	\[\cU_1\sigma (x,y)=\sigma(x,y+Ax), \quad \cU_2\sigma(x,y)=\sigma(x,B^{\top}y),\quad \widehat{\cU\sigma}(\xi,\eta)=e^{\pi i \xi \eta}\widehat{\sigma}(\xi,\eta).\] $\cU_1$ and $\cU_2$ are isomorphisms of $M^{\infty}_s(\rdd)$, as a consequence of the previous item. For what concerns $\cU$, an inspection of the proof of \cite[Cor.\ 14.5.5]{gro1 book} shows that any modulation space $M^{p,q}_s(\rdd)$ is invariant under the action of $\cU$. 
\end{proof}

We will also make use of the following easy result. 

\begin{lemma}
	\label{exp V0} Let $A$ be a Banach algebra of complex-valued functions on $\rd$ with respect to pointwise multiplication and assume $u \in A$. For any real $t$
	and integer $n\ge1$ we have
	\[
	e^{-i\frac{t}{n}u}=1+i\frac{t}{n}u_{0},
	\]
	where $u_{0}\in A$ and the following
	estimate holds: 
	\[
	\left\Vert u_{0}\right\Vert \leq\left\Vert u\right\Vert e^{|t|\left\Vert u\right\Vert }.
	\]
\end{lemma}
\begin{proof}
	We have 
	\[
	e^{-i\frac{t}{n}u}=\sum_{k=0}^{\infty}\left(-i\frac{t}{n}\right)^{k}\frac{u^{k}}{k!}=1+i\frac{t}{n}u_{0}
	\]
	with 
	\[
	u_0=- u \sum_{k=0}^{\infty}\left(-i\frac{t}{n}\right)^{k}\frac{u^{k}}{\left(k+1\right)!}.
	\]
	We can estimate the norm of $u_{0}$ as follows: 
	\begin{align*}
	\left\Vert u_{0}\right\Vert
	& \leq\left\Vert u\right\Vert \left(\sum_{k=0}^{\infty}\frac{|t|^{k}\left\Vert u\right\Vert^k}{\left(k+1\right)!}\right)\\
	& =\frac{1}{|t|}\left(e^{|t|\left\Vert u\right\Vert}-1\right)\leq\left\Vert u\right\Vert e^{|t|\left\Vert u\right\Vert}.
	\end{align*}
\end{proof}
Thanks to the subsequent result, we are able to treat Theorem \ref{maint sjo} as
a perturbation of Theorem \ref{maint minfty}. 
\begin{lemma}
	\label{sjo decomp}For any $\epsilon>0$ and $f\in M^{\infty,1}(\mathbb{R}^{d})$,
	there exist $f_{1}\in C^{\infty}_b(\mathbb{R}^{d})$ and
	$f_{2}\in M^{\infty,1}(\mathbb{R}^{d})$ such that 
	\[
	f=f_{1}+f_{2},\qquad\left\Vert f_{2}\right\Vert _{M^{\infty,1}}\le\epsilon.
	\]
\end{lemma}

\begin{proof}
	Fix $g\in\mathcal{S}(\mathbb{R}^{d})\backslash\left\{ 0\right\} $ with $\| g \|_{L^2}=1$, and set 
	\begin{equation} \label{f_1 decomp}
	f_{1} (y)= V_{g}^{*}\left(V_{g}f\cdot1_{A_{R}}\right)(y)= \int_{A_{R}}V_{g}f\left(x,\xi\right)e^{2\pi iy\xi}g\left(y-x\right)dx d\xi, \end{equation} 
	in the sense of distributions, where $1_{A_{R}}$ denotes the characteristic function of the set
	$A_{R}=\left\{ \left(x,\xi\right)\in\mathbb{R}^{2d}:\left|\xi\right|\le R\right\}$,
	and $R>0$ will be chosen later, depending on $\epsilon$.\par
	 The integral in \eqref{f_1 decomp} actually converges for every $y$ and defines a bounded function. Indeed, setting 
 	\[
	S\left(\xi\right)=\sup_{x\in\mathbb{R}^{d}}\left|V_{g}f\left(x,\xi\right)\right|	\] 	we have  $S\in L^{1}(\mathbb{R}^{d})$ by the assumption $f \in M^{\infty,1}(\rd)$, and for any $y \in \rd$,
\begin{align*}
	\left|f_{1}\left(y\right)\right| & =\left|\int_{A_{R}}V_{g}f\left(x,\xi\right)e^{2\pi iy\xi}g\left(y-x\right)dxd\xi \right|\\
	& \le\int_{A_{R}}\left|V_{g}f\left(x,\xi\right)\right|\left|g\left(y-x\right)\right| dxd\xi \\
	& \le\left(\int_{\rd}\left| g\left(y-x\right)\right|dx\right)\left(\int_{|\xi|\leq R}S\left(\xi\right)d\xi\right)\leq \|g\|_{L^1}\|S\|_{L^1}.
	\end{align*}
Similarly one shows that all the derivatives $\partial^\alpha f_1$ are bounded, using that $\xi^\alpha S(\xi)$ is integrable on $|\xi|\leq R$. Differentiation under the integral sign is permitted because for $y$ in a neighbourhood of any fixed $y_0\in\rd$ and every $N$,
\[
|V_{g}f\left(x,\xi\right)\partial_{y}^{\alpha}[e^{2\pi iy\xi}g\left(y-x\right)]|\leq C_N(1+|\xi|)^{|\alpha|}S(\xi)(1+|y_0-x|)^{-N},
\]
which is integrable in $A_R$. Hence $f_{1}\in C^{\infty}_b(\mathbb{R}^{d})$.\par

	Now, let
	\[
	f_{2}=f-f_{1}=V_{g}^{*}\left(V_{g}f\cdot1_{A_{R}^{c}}\right)
	\]
	where in the second equality we used the inversion formula for the STFT \eqref{inversion STFT}. The continuity of $V_g^* : L^{\infty,1}(\rdd) \to M^{\infty,1}(\rd)$ yields 
	
	\begin{align*}
	\left\Vert f_{2}\right\Vert _{M^{\infty,1}} & =\left\Vert V_{g}^{*}\left(V_{g}f\cdot1_{A_{R}^{c}}\right)\right\Vert _{M^{\infty,1}}\\
	& \lesssim\left\Vert V_{g}f\cdot1_{A_{R}^{c}}\right\Vert _{L^{\infty,1}}\\
	& = \int_{|\xi|> R} S(\xi) d\xi \le \epsilon
	\end{align*}
	provided that $R=R_{\epsilon}$ is large enough.
\end{proof}
As already claimed in the Introduction, we prove that the Sj\"ostrand class includes the Fourier transforms of (finite) complex measures. 
\begin{proposition}
	\label{fou meas sjo} Let $\cM (\rd)$ denote the space of complex Radon measures on $\rd$. The image of $\cM(\rd)$ under the Fourier transform is contained in $M^{\infty,1}(\rd)$, that is: $\cF\cM(\rd) \subset M^{\infty,1}(\rd). $
\end{proposition}
\begin{proof} 
	We regard $\mathcal{M}(\mathbb{R}^{d})\subset\mathcal{S}'(\mathbb{R}^{d})$. Therefore, for any non-zero window $g\in\mathcal{S}(\mathbb{R}^{d})$
	we can explicitly write the STFT of $\mu$:
	\[
	V_{g}\mu\left(x,\xi\right)=\left\langle \mu,M_{\xi}T_{x}g\right\rangle =\int_{\mathbb{R}^{d}}e^{-2\pi ix\xi}\overline{g\left(y-x\right)}d\mu(y).
	\]
	In view of the relation between the Wiener amalgam space $W^{p,q}(\rd)$ and $M^{p,q}(\rd)$, the claimed result is equivalent to prove that $\cM(\rd) \subset W^{\infty,1}(\rd)$. Indeed,
	\begin{flalign*}
	\left\Vert \mu\right\Vert _{W^{\infty,1}} & =\int_{\mathbb{R}^{d}} \sup_{\xi\in\mathbb{R}^{d}}\left|V_{g}\mu\left(x,\xi\right)\right| dx\\
	& \le\int_{\mathbb{R}^{d}} \sup_{\xi\in\mathbb{R}^{d}}\int_{\mathbb{R}^{d}}\big|e^{-2\pi ix\xi}\overline{g\left(y-x\right)}\big|d |\mu| \left(y\right) dx\\
	& =\int_{\mathbb{R}^{d}}\int_{\mathbb{R}^{d}}\big|{g\left(y-x\right)}\big|d|\mu|\left(y\right)dx\\
	& =\int_{\mathbb{R}^{d}}\int_{\mathbb{R}^{d}}\big|{g\left(y-x\right)}\big|dxd|\mu|\left(y\right)\\
	& =\left\Vert g\right\Vert _{L^{1}}|\mu|(\mathbb{R}^{d}).
	\end{flalign*}
\end{proof}

\section{Proof of Theorem \ref{maint minfty} and Corollary \ref{maint s000}}
\subsection{Proof of Theorem \ref{maint minfty}}
Recall that $H_{0}=a^{\text{w}}$ is the Weyl quantization of the real
quadratic form $a\left(x,\xi\right)$ on $\mathbb{R}^{2d}$ and we are assuming
$V\in M^{\infty}_s (\rd)$, with $s>2d$ (the multiplication by $V$ coincides with $\sigma_V^{\mathrm{w}}$, as discussed in Remark \ref{mult symbol}). The proof will be carried on for $t>0$, since the case $t<0$ is similar. Actually, the upper-right block of the matrix $\cA_{-t} = \cA_t^{-1}$ is $-B_t^{\top}$ (cf.\ \cite[Eq.\ (2.6)]{dG symp met}), hence $\det B_t \ne 0$ if and only if  $\det B_{-t} \ne 0$. 

 Having in mind the framework outlined in the introductory Section \ref{sec maint}, we start from Trotter formula \eqref{trotter maint}. We employ Lemma \ref{exp V0} and the notation $e^{-itH_0}=\mu(\mathcal{A}_t)$ from Section \ref{mataop} in order to write
\[
E_{n}\left(t\right)=\left(e^{-i\frac{t}{n}H_0}e^{-i\frac{t}{n}V}\right)^{n}=\left(\mu\left(\mathcal{A}_{t/n}\right)\left(1+i\frac{t}{n}V_{0}\right)\right)^{n}
\]
for a suitable $V_0=V_{0,n,t}$. According to Remark \ref{mult symbol}, we identify $1+i\frac{t}{n}V_{0}$
with the Weyl operator with symbol $1+i\frac{t}{n}\sigma_{V_0}$, where $\sigma_{V_{0}} = V_{0}\otimes1$. By the assumption $V\in M^\infty_s(\rd)$, Remark \ref{mult symbol} and Lemma \ref{exp V0} we have
\begin{equation}\label{eqv0}
\|\sigma_{V_0}\|_{M_{s}^{\infty}}\leq C(t)
\end{equation}
for some constant $C(t)>0$ independent of $n$. 
By applying \eqref{symp cov} repeatedly, the ordered product of operators in $E_n(t)$ can be expanded as
\[ 
E_n(t) = \left[ \prod_{k=1}^{n} \left( I+i\frac{t}{n} \left(\sigma_{V_{0}} \circ \cA_{-k\frac{t}{n}} \right)^{\mathrm{w}} \right) \right] \mu\left(\cA_{t/n}\right)^n = a_{n,t}^{\text{w}}\,\mu(\mathcal{A}_{t}),
\]
where, for any $t$ and $n\ge1$:
\begin{align*}
\left\Vert a_{n,t}\right\Vert _{M_{s}^{\infty}} & = \bigg\Vert \prod_{k=1}^{n} \left(1+i\frac{t}{n} \left(\sigma_{V_{0}} \circ \cA_{-k\frac{t}{n}} \right) \right)\bigg\Vert_{M_{s}^{\infty}} \\
& \leq \prod_{k=1}^{n} \left( 1+\frac{t}{n} \big\Vert \sigma_{V_{0}} \circ \cA_{-k\frac{t}{n}} \big\Vert_{M^{\infty}_s} \right),
\end{align*}
where in the first product symbol we mean the twisted product $\#$ of symbols - cf.\ Section \ref{weyl sec}. \par
By Lemma \ref{tech lem} applied with $T=t$ and \eqref{eqv0},
we then have 
\begin{equation}\label{aennet}
\left\Vert a_{n,t}\right\Vert _{M_{s}^{\infty}}\le\left(1+\frac{t}{n}C\left(t\right)\right)^{n}\le e^{C\left(t\right)t},
\end{equation}
for some new locally bounded constant $C(t)>0$ independent of $n$. 

Since $\mathcal{A}_{t}$ is a free symplectic matrix by assumption,
by \eqref{met int formula} and  \eqref{swap weyl fio} we explicitly have 
\begin{alignat*}{1}
E_{n}\left(t\right)\psi(x) & =a_{n,t}^{\text{w}}\,\mu\left(\mathcal{A}_{t}\right)\psi(x)\\
& =c(t)\left| \det B_{t} \right|^{-1/2}\int_{\mathbb{R}^{d}}e^{2\pi i\Phi_{t}\left(x,y\right)}\widetilde{a_{n,t}}\left(x,y\right)\psi\left(y\right)dy,
\end{alignat*}
where $\Phi_t$ is given in \eqref{phit} and $c(t)$ is a suitable complex factor of modulus $1$.\par 
 Therefore,
we managed to write $E_{n}\left(t\right)$ as an integral operator
with kernel 
\[
e_{n,t}\left(x,y\right)=c(t)\left|\det B_{t}\right|^{-1/2}e^{2\pi i\Phi_{t}\left(x,y\right)}\widetilde{a_{n,t}}(x,y),
\]
Now, consider the integral kernel $u_t$
of the propagator $U\left(t\right)=e^{-it\left(H_{0}+V\right)}$ and define for consistency $\widetilde{a_{t}}\in\mathcal{S}'(\mathbb{R}^{2d})$
such that 
\[
u_{t}\left(x,y\right)= c(t)\left| \det B_{t} \right|^{-1/2}e^{2\pi i\Phi_{t}\left(x,y\right)}\widetilde{a_{t}}(x,y).
\]
Since we know by the usual Trotter formula \eqref{trotter maint} that for any fixed $t$ 
\[
\| E_{n}(t)\psi-U(t)\psi\|_{L^2}\to 0 ,\qquad\forall\psi\in L^2(\mathbb{R}^{d}),
\]
we have $E_n(t)\rightarrow U(t)$ in $\cL_s(\cS(\rd),\cS'(\rd))$, because $\cS(\rd)\hookrightarrow L^2(\rd) \hookrightarrow \cS'(\rd)$. As a consequence of Proposition \ref{propokernel}, we get $e_{n,t}\rightarrow u_t$ in $\cS'(\rd)$. This is equivalent to 
\begin{equation}\label{16bis}
\widetilde{a_{n,t}}\rightarrow \widetilde{a_{t}}\quad{\rm in}\ \mathcal{S}'(\mathbb{R}^{2d}).
\end{equation}
Therefore, for any
non-zero $\Psi\in\mathcal{S}(\mathbb{R}^{2d})$ we have pointwise
convergence of the corresponding short-time Fourier transforms: for any
fixed $\left(z,\zeta\right)\in\mathbb{R}^{4d}$,
\begin{equation}\label{stft point conv}
V_{\Psi}\widetilde{a_{n,t}}\left(z,\zeta\right)=\left\langle \widetilde{a_{n,t}},M_{\zeta}T_{z}\Psi\right\rangle \rightarrow\left\langle \widetilde{a_{t}},M_{\zeta}T_{z}\Psi\right\rangle =V_{\Psi}\widetilde{a_{t}}(z,\zeta).
\end{equation}
By  \eqref{aennet} and Lemma \ref{tech lem} we see that the sequence $\widetilde{a_{n,t}}$, for any fixed $t$, is bounded in $M_{s}^{\infty}(\rdd)$. Hence,
there exists a constant $C=C(t)$ independent of $n$ such that
\begin{equation}\label{17eq}
\left|V_{\Psi}\widetilde{a_{n,t}}\left(z,\zeta\right)\right|\le C\left\langle \zeta\right\rangle ^{-s},\qquad\forall z,\zeta\in\mathbb{R}^{2d}.
\end{equation}
Combining this estimate with \eqref{stft point conv} immediately yields $\widetilde{a_{t}} \in M^{\infty}_s(\rdd)$ as well, hence the first claim of Theorem \ref{maint minfty}. 

For the remaining part, we argue as follows: choose a non-zero window $\Psi\in C_{c}^{\infty}(\mathbb{R}^{2d})$
and set $\Theta\in C_{c}^{\infty}(\mathbb{R}^{2d})$ with
$\Theta=1$ on $\mathrm{supp}\Psi$; for any fixed $z\in\mathbb{R}^{2d}$
and $0<r<s-2d$, we have
\begin{align*}
\left\Vert  \mathcal{F}\left[\left(e_{n,t}-u_{t}\right)\overline{T_{z}\Psi}\right]\right\Vert _{L_{r}^{1}} & =\left| \det B_{t} \right|^{-1/2}\left\Vert \mathcal{F}\left[e^{2\pi i\Phi_{t}}\left(\widetilde{a_{n,t}}-\widetilde{a_{t}}\right)\overline{T_{z}\Psi}\right]\right\Vert _{L_{r}^{1}}\\
& =\left| \det B_{t} \right|^{-1/2}\left\Vert \mathcal{F}\left[\left(T_{z}\Theta e^{2\pi i\Phi_{t}}\right)\left(\widetilde{a_{n,t}}-\widetilde{a_{t}}\right)\overline{T_{z}\Psi}\right]\right\Vert _{L_{r}^{1}}\\
& =\left| \det B_{t} \right|^{-1/2}\left\Vert \mathcal{F}\left[T_{z}\Theta e^{2\pi i\Phi_{t}}\right]*\mathcal{F}\left[\left(\widetilde{a_{n,t}}-\widetilde{a_{t}}\right)\overline{T_{z}\Psi}\right]\right\Vert _{L_{r}^{1}}\\
& \lesssim\left| \det B_{t} \right|^{-1/2}\left\Vert \mathcal{F}\left[T_{z}\Theta e^{2\pi i\Phi_{t}}\right]\right\Vert _{L_{r}^{1}}\left\Vert \mathcal{F}\left[\left(\widetilde{a_{n,t}}-\widetilde{a_{t}}\right)\overline{T_{z}\Psi}\right]\right\Vert _{L_{r}^{1}},
\end{align*}
where  the convolution inequality in the last step is an easy consequence of \textit{Peetre's inequality}\footnote{Namely, $ \la x-y \ra^r \le C_r \la x \ra^r \la y \ra^r$, for $ x,y\in \rd, r\geq 0$.}.

Clearly, $T_{z}\Theta e^{2\pi i\Psi_{t}}\in C_{c}^{\infty}(\mathbb{R}^{2d})$,
while \[ \left\Vert \mathcal{F}\left[\left(\widetilde{a_{n,t}}-\widetilde{a_{t}}\right)\overline{T_{z}\Psi}\right]\right\Vert _{L_{r}^{1}}\rightarrow0 \]
by dominated convergence, using \eqref{stft point conv} and
\begin{align*}
|\mathcal{F}\left[\left(\widetilde{a_{n,t}}-\widetilde{a_{t}}\right)\overline{T_{z}\Psi}\right]| \langle \zeta\rangle ^{r}
& =\left|V_{\Psi}\left(\widetilde{a_{n,t}}-\widetilde{a_{t}}\right)\left(z,\zeta\right)\right|\left\langle \zeta\right\rangle ^{r} \\
&  \leq C\left\langle \zeta\right\rangle ^{r-s}\in L^1(\rdd),
\end{align*}
because $s-r>2d$, where in the last inequality we used \eqref{17eq} and the fact that $\widetilde{a_{t}}\in M^\infty_s(\rdd)$.\par
This gives the claimed convergence in $\left(\mathcal{F}L_{r}^{1}\right)_{\mathrm{loc}}(\rdd)$.\par
To conclude, we have that 
\[
\left\Vert \left(e_{n,t}-u_{t}\right)\overline{T_{z}\Psi}\right\Vert _{L^{\infty}}\le\left\Vert V_{\Psi}\left(e_{n,t}-u_{t}\right)\left(z,\cdot\right)\right\Vert _{L^{1}}\rightarrow0,
\]
and in particular this yields uniform convergence on compact subsets: for any compact $K\subset\mathbb{R}^{2d}$, choose $\Psi\in\mathcal{S}(\mathbb{R}^{2d})$,
$\Psi=1$ on $K$. 
\subsection{Proof of Corollary \ref{maint s000}}
The proof of Corollary \ref{maint s000} is then immediate, since $C^{\infty}_b(\mathbb{R}^{2d})=\bigcap_{s\ge0}M_{s}^{\infty}(\mathbb{R}^{2d})$ and $C^{\infty}(\rdd) = \bigcap_{r>0} \left( \cF L^1_r \right)_{\mathrm{loc}}(\rdd)$; we leave the easy proof of the latter equality to the interested reader. 

\section{Proof of Theorem \ref{maint sjo}}
We now assume $V\in M^{\infty,1}(\mathbb{R}^{d})$. Therefore for an arbitrary $\epsilon>0$, Lemma
\ref{sjo decomp} allows us to write $V=V_{1}+V_{2}$, with $V_{1}\in C^{\infty}_b(\mathbb{R}^{d})$
and $V_{2}\in M^{\infty,1}(\mathbb{R}^{d})$ with $\left\Vert V_{2}\right\Vert _{M^{\infty,1}}\le \epsilon$ and clearly
\[
\left\Vert V_{1}\right\Vert _{M^{\infty,1}}\leq\left\Vert V\right\Vert _{M^{\infty,1}}+\left\Vert V_{2}\right\Vert _{M^{\infty,1}}\leq\left\Vert V\right\Vert _{M^{\infty,1}}+\epsilon \le 1+ \|V\|_{M^{\infty,1}},
\]
assuming, from now on, $\epsilon \leq 1$. \par
Notice that 
\begin{align*}
e^{-i\frac{t}{n}\left(V_{1}+V_{2}\right)} & =1+\sum_{k=1}^{\infty}\frac{1}{k!}\left(-i\frac{t}{n}\right)^{k}\left(V_{1}+V_{2}\right)^{k}\\
% & =1-i\frac{t}{n}\left[V_{1}\sum_{k=1}^{\infty}\frac{1}{k!}\left(-i\frac{t}{n}\right)^{k-1}V_{1}^{k-1}\right] \\ & - i\frac{t}{n}\left[V_{2}\sum_{k=1}^{\infty}\frac{1}{k!}\left(-i\frac{t}{n}\right)^{k-1}\left(\sum_{j=0}^{k-1}V_{1}^{k-1-j}V_{2}^{j}\right)\right]\\
& =1+i\frac{t}{n}V'_{1}+i\frac{t}{n}V'_{2},
\end{align*}
where we set
\[ V'_{1} = -V_{1}\sum_{k=1}^{\infty}\frac{1}{k!}\left(-i\frac{t}{n}\right)^{k-1}V_{1}^{k-1},  \]
\[ V'_{2} = -\sum_{k=1}^{\infty}\frac{1}{k!}\left(-i\frac{t}{n}\right)^{k-1}((V_{1}+V_{2})^k-V_1^k). \]
Now, fix once for all any $s>2d$. The norms of $V'_{1}=V'_{1,n,t}$ and
$V'_{2}=V'_{2,n,t}$ can be estimated as follows for any $t>0$ (cf.\ the proof of Lemma \ref{exp V0}). We have
\begin{equation}\label{c1t}
\left\Vert V'_{1}\right\Vert _{M^{\infty,1}}\le\left\Vert V_{1}\right\Vert _{M^{\infty,1}}e^{t\left\Vert V_{1}\right\Vert _{M^{\infty,1}}}\le\left(1+\left\Vert V\right\Vert _{M^{\infty,1}}\right)e^{t\left(1+\left\Vert V\right\Vert _{M^{\infty,1}}\right)}=: C_1(t),
\end{equation}
\begin{equation}\label{c2t}
\left\Vert V'_{1}\right\Vert _{M_{s}^{\infty}}\le\left\Vert V_{1}\right\Vert _{M_{s}^{\infty}}e^{t\left\Vert V_{1}\right\Vert _{M_{s}^{\infty}}}=: C_2(t,\epsilon).
\end{equation} 
Similarly, using the elementary inequality 
\[
(a+b)^k-a^k\leq kb(a+b)^{k-1},\qquad a,b\geq 0,\ k\geq 1,
\]
we obtain
\begin{equation}\label{c3t}
\left\Vert V'_{2}\right\Vert _{M^{\infty,1}}\le\left\Vert V_{2}\right\Vert _{M^{\infty,1}}e^{t\left(\left\Vert V_{1}\right\Vert _{M^{\infty,1}}+\left\Vert V_{2}\right\Vert _{M^{\infty,1}}\right)}\le \epsilon e^{t\left(2+\left\Vert V\right\Vert _{M^{\infty,1}}\right)} =:  \epsilon\, C_3(t).
\end{equation}
Here $C_1(t)$ and $C_3(t)$ are independent of $n$ and $\epsilon$ and $C_2(t,\epsilon)$ is independent of $n$. The approximate propagator $E_n(t)$ thus becomes
\begin{align*}
E_{n}\left(t\right) & =\left(e^{-i\frac{t}{n}H_0}e^{-i\frac{t}{n}\left(V_{1}+V_{2}\right)}\right)^{n}\\
& =\left(\mu\left(\mathcal{A}_{t/n}\right)\left(1+i\frac{t}{n}V'_{1}+i\frac{t}{n}V'_{2}\right)\right)^{n},
\end{align*}
and similar arguments to those of the previous section yield 
\begin{align}\label{proden}
E_{n}\left(t\right) & =\left[\prod_{k=1}^{n}\left(I+i\frac{t}{n}\left(\sigma_{V'_{1}}\circ\mathcal{A}_{-k\frac{t}{n}}\right)^{\text{w}}+i\frac{t}{n}\left(\sigma_{V'_{2}}\circ\mathcal{A}_{-k\frac{t}{n}}\right)^{\text{w}}\right)\right]\left(\mu\left(\mathcal{A}_{t/n}\right)\right)^{n}\\
& =:\left[a_{n,t}^{\text{w}}+b_{n,t}^{\text{w}}\right]\mu(\mathcal{A}_{t})\nonumber,
\end{align}
where we set \[a_{n,t} = \prod_{k=1}^{n}\left(1+i\frac{t}{n}\left(\sigma_{V'_{1}}\circ\mathcal{A}_{-k\frac{t}{n}}\right)\right),\] 
and in the latter product we mean the twisted product $\#$ of symbols. \par
The term $a_{n,t}^{\text{w}}$ can be estimated as in the proof of Theorem \ref{maint minfty}; in particular, using \eqref{c2t}, we get 
\begin{equation}\label{antest}
\left\Vert a_{n,t}\right\Vert _{M_{s}^{\infty}}\leq C(t,\epsilon)
\end{equation}
 cf.\ \eqref{aennet}. \par
 In order to estimate the $M^{\infty,1}$ norm of the remainder $b_{n,t}$, it is useful the following result, which can be easily proved by induction on $n$. 
\begin{lemma}\label{lem alg sum}
	Let $A$ be a Banach algebra. For any $u_{1},\ldots,u_{n},v_{1},\ldots,v_{n}\in A$,
	with $\left\Vert u_{i}\right\Vert \leq R$ and $\left\Vert v_{i}\right\Vert \leq S$
	for any $i=1,\ldots,n$ and some $R,S>0$, and setting $w_{k}=u_{k}+v_{k}$, we have 
	\[
	\prod_{k=1}^{n}\left(u_{k}+v_{k}\right)=u_{1}u_{2}\ldots u_{n}+z_{n},
	\]
	where 
	\begin{align*}
	z_{n} & =v_{1}w_{2}\ldots w_{n}+u_{1}v_{2}w_{3}\ldots w_{n}+\ldots+u_{1}u_{2}\ldots u_{n-2}v_{n-1}w_{n}+u_{1}u_{2}\ldots u_{n-1}v_{n},
	\end{align*}
	and therefore
	\[
	\left\Vert z_{n}\right\Vert \leq nS\left(R+S\right)^{n-1}.
	\]
	
\end{lemma}

Setting 
\[
u_{k}=1+i\frac{t}{n}\left(\sigma_{V'_{1}}\circ\mathcal{A}_{-k\frac{t}{n}}\right),\qquad v_{k}=i\frac{t}{n}\left(\sigma_{V'_{2}}\circ\mathcal{A}_{-k\frac{t}{n}}\right),\quad k=1,\ldots,n,
\] 
and applying Lemma \ref{tech lem} with $T=t$, and  \eqref{c1t} and \eqref{c3t}, we get
\[
\Vert u_{k} \Vert _{M^{\infty,1}}=\big\Vert 1+i\frac{t}{n}\left( \sigma_{V'_1}\circ\mathcal{A}_{-k\frac{t}{n}} \right) \big\Vert _{M^{\infty,1}}\leq1+\frac{t}{n}\big\Vert \sigma_{V'_{1}}\circ\mathcal{A}_{-k\frac{t}{n}}\big\Vert _{M^{\infty,1}}\leq1+\frac{t}{n}C(t),
\]
\[
\left\Vert v_{k}\right\Vert _{M^{\infty,1}}=\frac{t}{n}\big\Vert \sigma_{V'_{2}}\circ\mathcal{A}_{-k\frac{t}{n}}\big\Vert _{M^{\infty,1}}\leq\frac{t}{n}C\left(t\right)\epsilon,
\]
for some locally bounded constant $C(t)>0$ independent of $n$ and $\epsilon$.
Therefore, by Lemma \ref{lem alg sum},
\begin{equation}\label{stima bnt}
\left\Vert b_{n,t}\right\Vert _{M^{\infty,1}}\le n\frac{t}{n}C\left(t\right)\epsilon\Big(1+2\frac{t}{n}C\left(t\right)\Big)^{n-1}\le\epsilon tC\left(t\right)e^{2tC\left(t\right)}.
\end{equation}
Following the pathway of the proof of Theorem \ref{maint minfty}, we write $E_{n}\left(t\right)$
as an integral operator with kernel 
\begin{align*}
e_{n,t}(x,y) & =c(t)\left| \det B_{t} \right|^{-1/2}e^{2\pi i\Phi_{t}\left(x,y\right)}\big(\widetilde{a_{n,t}}+\widetilde{b_{n,t}}\big)\left(x,y\right) \\ & =c(t)\left| \det B_{t} \right|^{-1/2}e^{2\pi i\Phi_{t}\left(x,y\right)}k_{n,t}(x,y),
\end{align*}
that is $k_{n,t}=\widetilde{a_{n,t}}+\widetilde{b_{n,t}}$, and the Trotter formula \eqref{trotter maint} combined with Proposition \ref{propokernel} imply that $k_{n,t} \rightarrow k_t$ in $\mathcal{S}'(\mathbb{R}^{2d})$, where the distribution $k_t$ is conveniently introduced to rephrase the integral kernel $u_t$ of the propagator $U\left(t\right)=e^{-it\left(H_{0}+V\right)}$ as
\[
u_{t}\left(x,y\right)= c(t)\left| \det B_{t} \right|^{-1/2}e^{2\pi i\Phi_{t}\left(x,y\right)}k_t(x,y).
\] 
By repeating this argument with $V_2=0$ (hence $\widetilde{b_{n,t}}=0$ and $k_{n,t}=\widetilde{a_{n,t}}$) we see that $\widetilde{a_{n,t}}$ converges in $\cS'(\rdd)$ as well, hence $\widetilde{b_{n,t}}$ converges in $\mathcal{S}'(\mathbb{R}^{2d})$
by difference. Therefore, for any non-zero $\Psi\in\mathcal{S}(\mathbb{R}^{2d})$ the functions $V_{\Psi}\widetilde{a_{n,t}}$ and $V_{\Psi}\widetilde{b_{n,t}}$ converge pointwise in $\R^{4d}$.\par
We need a technical lemma at this point.
\begin{lemma}
	\label{lem limsup}Let $F_n$ and $G_n$
	be two sequences of complex-valued functions on $\mathbb{R}^{2d}$
	such that $F_{n}\rightarrow F$, $G_{n}\rightarrow G$ pointwise, and assume $\left|F_{n}\right|\le H\in L^{1}(\mathbb{R}^{2d})$
	and $\left\Vert G_{n}\right\Vert _{L^{1}}\le\epsilon$ for any $n\in\mathbb{N}$.
	Then, 
	\[
	\limsup_{n\rightarrow\infty}\left\Vert F_{n}+G_{n}-\left(F+G\right)\right\Vert _{L^{1}}\le2\epsilon. 	\]
\end{lemma}

\begin{proof}
	First, notice that $\left\Vert G\right\Vert _{L^{1}}\le\epsilon$
	by Fatou's lemma. Now, 
	\[
	\left\Vert F_{n}+G_{n}-\left(F+G\right)\right\Vert _{L^{1}}\le\left\Vert F_{n}-F\right\Vert _{L^{1}}+\left\Vert G_{n}-G\right\Vert _{L^{1}},
	\]
	where the first term on the right-hand side goes to zero by dominated convergence,
	while for the other one we have $\left\Vert G_{n}-G\right\Vert _{L^{1}}\le2\epsilon$. The desired conclusion is then immediate.
	 \end{proof}
For any fixed $z\in\mathbb{R}^{2d}$, set $F_{n}(\zeta)=V_{\Psi} \widetilde{a_{n,t}}(z,\zeta)$
and $G_{n}(\zeta)=V_{\Psi}\widetilde{b_{n,t}}(z,\zeta)$. \par
By Lemma \ref{tech lem}  and  \eqref{antest} we have 
\[
 \sup_{\zeta\in\rdd}\left\langle \zeta\right\rangle ^{s}\left|F_{n}\left(\zeta\right)\right|\lesssim \| \widetilde{a_{n,t}}\|_{M^\infty_s} \lesssim \|{a_{n,t}}\|_{M^\infty_s} \leq C(t,\epsilon). 
 \]
Similarly, by Lemma \ref{tech lem} and \eqref{stima bnt},
\[
\left\Vert G_{n}\right\Vert _{L^{1}} \lesssim  \Vert \widetilde{b_{n,t}}\Vert _{M^{\infty,1}}\lesssim\left\Vert b_{n,t}\right\Vert _{M^{\infty,1}}\leq\epsilon\,  C(t).
\]
These estimates yield two results: on the one hand, the first claim of Theorem \ref{maint sjo} is proved. 
On the other hand, the assumptions of Lemma
\ref{lem limsup} are satisfied: we have $\left(F_n+G_n\right)(\zeta)=V_{\Psi}k_{n,t}(z,\zeta)$ and $\left(F+G\right)(\zeta)=V_{\Psi}k_t(z,\zeta)$, and therefore we obtain
\[
\limsup_{n\rightarrow\infty}\left\Vert \mathcal{F}\left[\left(k_{n,t}-k_{t}\right)\overline{T_{z}\Psi}\right]\right\Vert _{L^{1}} \le 2 \epsilon\, C(t).
\]
Since $\epsilon$ can be made arbitrarily small and the left-hand side is independent of $\epsilon$, we conclude that
\[
\lim_{n\rightarrow\infty}\left\Vert \mathcal{F}\left[\left(k_{n,t}-k_{t}\right)\overline{T_{z}\Psi}\right]\right\Vert _{L^{1}}=0,
\]
in particular $k_{n,t}\rightarrow k_{t}$ in $(\mathcal{F}L^{1})_{\mathrm{loc}}$$(\mathbb{R}^{2d})$.\par
Finally, with the help of a suitable bump function $\Theta$ as
in the preceding section, for any fixed $z\in\mathbb{R}^{2d}$ we deduce
\[
\left\Vert  \mathcal{F}\left[\left(e_{n,t}-u_{t}\right)\overline{T_{z}\Psi}\right]\right\Vert _{L^{1}} \le \left| \det B_{t} \right|^{-1/2}\left\Vert \mathcal{F}\left[\left(T_{z}\Theta e^{2\pi i\Phi_{t}}\right)\right]\right\Vert _{L^{1}}\left\Vert \mathcal{F}\left[\left(k_{n,t}-k_{t}\right)\overline{T_{z}\Psi}\right]\right\Vert _{L^{1}},
\]
and thus
\[
\left\Vert  \mathcal{F}\left[\left(e_{n,t}-u_{t}\right)\overline{T_{z}\Psi}\right]\right\Vert _{L^{1}} \rightarrow0.
\]
This gives $e_{n,t}\rightarrow u_{t}$ in $(\mathcal{F}L^{1})_{\mathrm{loc}}$$(\mathbb{R}^{2d})$ and therefore uniformly on compact subsets of $\rdd$.

\end{document}